\documentclass[11pt]{article}
\usepackage{fullpage}

\usepackage{amsmath,amsfonts,amssymb,amsthm,fixmath}
\usepackage{graphicx}
\usepackage{enumerate}
\usepackage{bbm}
\usepackage{verbatim}
\usepackage{hyperref,color}
\usepackage{cleveref}

\allowdisplaybreaks 

\theoremstyle{plain}
\newtheorem{thm}{Theorem}[section]
\newtheorem{theorem}[thm]{Theorem}
\newtheorem{lemma}[thm]{Lemma}

\newtheorem{corollary}[thm]{Corollary}

\newtheorem{claim}[thm]{Claim}

\theoremstyle{definition}

\newtheorem{definition}[thm]{Definition}

\newtheorem*{ass*}{Assumption}
\newtheorem*{assumption*}{Assumption}
\theoremstyle{remark}

\newtheorem{remark}[thm]{Remark}

\crefname{section}{Section}{Sections}
\crefname{lemma}{Lemma}{Lemmas}
\crefname{theorem}{Theorem}{Theorems}
\crefname{appendix}{Appendix}{Appendices}
\crefname{defn}{Definition}{Definitions}
\crefname{conjecture}{Conjecture}{Conjectures}
\crefname{definition}{Definition}{Definitions}
\crefname{fact}{Fact}{Facts}
\crefname{fact}{Fact}{Facts}
\crefname{clm}{Claim}{Claims}
\crefname{claim}{Claim}{Claims}
\crefname{prop}{Proposition}{Propositions}
\crefname{proposition}{Proposition}{Propositions}
\crefname{algocf}{Algorithm}{Algorithms}

\newcommand{\E}{\mathbb{E}}
\newcommand{\Exp}{\E}

\newcommand{\Var}{\mathrm{Var}}
\newcommand{\Prob}{\mathrm{Prob}}

\newcommand{\poly}{\mathrm{poly}}
\newcommand{\dist}{\mathrm{dist}}

\newcommand{\R}{\mathbb{R}}

\definecolor{ao(english)}{rgb}{0.0, 0.5, 0.0}
\definecolor{airforceblue}{rgb}{0.36, 0.54, 0.66}
\definecolor{amber}{rgb}{1.0, 0.75, 0.0}
\definecolor{amber(sae/ece)}{rgb}{1.0, 0.49, 0.0}
\definecolor{amethyst}{rgb}{0.6, 0.4, 0.8}
\definecolor{azure(colorwheel)}{rgb}{0.0, 0.5, 1.0}
\definecolor{bittersweet}{rgb}{1.0, 0.44, 0.37}
\definecolor{bleudefrance}{rgb}{0.19, 0.55, 0.91}
\definecolor{blue(pigment)}{rgb}{0.2, 0.2, 0.6}
\definecolor{blue-violet}{rgb}{0.54, 0.17, 0.89}
\definecolor{britishracinggreen}{rgb}{0.0, 0.26, 0.15}
\definecolor{byzantine}{rgb}{0.74, 0.2, 0.64}
\definecolor{byzantium}{rgb}{0.44, 0.16, 0.39}
\definecolor{cadmiumgreen}{rgb}{0.0, 0.42, 0.24}
\definecolor{darkmagenta}{rgb}{0.55, 0.0, 0.55}
\definecolor{darkmidnightblue}{rgb}{0.0, 0.2, 0.4}
\definecolor{darkpastelpurple}{rgb}{0.59, 0.44, 0.84}
\definecolor{darkspringgreen}{rgb}{0.09, 0.45, 0.27}
\definecolor{deeppink}{rgb}{1.0, 0.08, 0.58}
\definecolor{dollarbill}{rgb}{0.52, 0.73, 0.4}
\definecolor{electricviolet}{rgb}{0.56, 0.0, 1.0}
\definecolor{ferrarired}{rgb}{1.0, 0.11, 0.0}
\definecolor{forestgreen(traditional)}{rgb}{0.0, 0.27, 0.13}
\definecolor{goldenbrown}{rgb}{0.6, 0.4, 0.08}
\definecolor{glaucous}{rgb}{0.38, 0.51, 0.71} 
\definecolor{heliotrope}{rgb}{0.87, 0.45, 1.0}
\definecolor{iris}{rgb}{0.35, 0.31, 0.81}

\newcommand{\ra}{\rightarrow}

\newcommand{\Trelax}{T_{\text{relax}}}
\newcommand{\Tmix}{T_{\text{mix}}}

\title{Optimal Mixing via Tensorization for Random Independent Sets on Arbitrary Trees}

\author{
 Charilaos Efthymiou\thanks{Department of Computer Science, University of Warwick, UK. Email: charilaos.efthymiou@warwick.ac.uk. Research supported by EPSRC NIA, grant EP/V050842/1,  and  Centre of Discrete Mathematics and Applications (DIMAP).}
 \and
 Thomas P. Hayes\thanks{Dept. of Computer Science and Engineering, University at Buffalo.  Email: hayest@gmail.com.}
 \and
 	Daniel \v{S}tefankovi\v{c}\thanks{Department of Computer Science, University of Rochester. Email: stefanko@cs.rochester.edu. Research supported in part by NSF grant CCF-1563757.}
  	\and Eric Vigoda\thanks{Department of Computer Science, University of California, Santa Barbara. Email: vigoda@ucsb.edu. 
	Research supported in part by NSF grant CCF-2147094.}
}

\date{\today}

\begin{document}

\maketitle

\begin{abstract}
We study the mixing time of the single-site update Markov chain, known as the Glauber dynamics, for generating a random independent set of a tree.  Our focus is obtaining optimal convergence results for arbitrary trees.  We consider the more general problem of sampling from the Gibbs distribution in the hard-core model where independent sets are weighted by a parameter $\lambda>0$; the special case $\lambda=1$ corresponds to the uniform distribution over all independent sets.  Previous work of Martinelli, Sinclair and Weitz (2004) obtained optimal mixing time bounds for the complete $\Delta$-regular tree for all $\lambda$.  
However, Restrepo, Stefankovic, Vera, Vigoda, and Yang (2014) showed that for sufficiently large $\lambda$ there are bounded-degree trees where optimal mixing does not hold.
Recent work of Eppstein and Frishberg (2022) proved a polynomial mixing time bound for the Glauber dynamics for arbitrary trees, and more generally for graphs of bounded tree-width.

We establish an optimal bound on the relaxation time (i.e., inverse spectral gap) of $O(n)$ for the Glauber dynamics for unweighted independent sets on arbitrary trees.  
We stress that our results hold for arbitrary trees and there is no dependence on the maximum degree $\Delta$.  Interestingly, our results extend (far) beyond the uniqueness threshold which is on the order $\lambda=O(1/\Delta)$.  Our proof approach is inspired by recent work on spectral independence.  In fact, we prove that spectral independence holds with a constant independent of the maximum degree for any tree, but this does not imply mixing for general trees as the optimal mixing results of Chen, Liu, and Vigoda (2021) only apply for bounded degree graphs.  We instead utilize the combinatorial nature of independent sets to directly prove approximate tensorization of variance via a non-trivial inductive proof.

\end{abstract}

\section{Introduction}

This paper studies the mixing time of the Glauber dynamics for the  hard-core model assuming that the underlying graph is an {\em arbitrary} tree.
In the hard-core model, we are given a graph $G=(V,E)$ and an activity $\lambda>0$.  The model is defined on the collection of all independent sets of $G$ (regardless of size), which we denote as $\Omega$.
 Each independent set $\sigma\in\Omega$ is assigned a weight $w(\sigma)=\lambda^{|\sigma|}$ where $|\sigma|$ is the number of vertices contained in the independent set $\sigma$.   
 The Gibbs distribution $\mu$ is defined on~$\Omega$: for $\sigma\in\Omega$, let $\mu(\sigma)=w(\sigma)/Z$ where $Z=\sum_{\tau\in\Omega} w(\tau)$ is known as the partition function.  When $\lambda=1$ then every independent set has weight one and hence the Gibbs distribution $\mu$ is the uniform distribution over (unweighted) independent sets.

 Our goal is to sample from $\mu$ (or estimate $Z$) in time polynomial in $n=|V|$.  Our focus is on trees.  These sampling and counting problems are computationally easy on trees using dynamic programming algorithms.  Nevertheless, our interest is to understand the convergence properties of a simple Markov Chain Monte Carlo (MCMC) algorithm known as the Glauber dynamics for sampling from the Gibbs distribution.

The Glauber dynamics (also known as the Gibbs sampler) is the simple single-site update Markov chain for sampling from the Gibbs distribution of a graphical model.   For the hard-core model with activity $\lambda$, the transitions $X_t\rightarrow X_{t+1}$ of the Glauber dynamics are defined as follows: first, choose a random  vertex $v$.  Then, with probability $\frac{\lambda}{1+\lambda}$ set $X'=X_t\cup\{v\}$ and with the complementary probability set $X'=X_t\setminus \{v\}$.  If  $X'$ is an independent set,  then  set $X_{t+1}=X'$ and otherwise set $X_{t+1}=X_t$.

We consider two standard notions of convergence to stationarity.  The {\em relaxation time} is the inverse spectral gap, i.e., $(1-\lambda^*)^{-1}$ where $\lambda^*=\max\{\lambda_2,|\lambda_N|\}$ and $1 = \lambda_1 > \lambda_2\geq\dots\geq\lambda_N>-1$ are the eigenvalues of the transition matrix $P$ for the Glauber dynamics.  The relaxation time is a key quantity in the running time for approximate counting algorithms (see, e.g., {\v{S}}tefankovi\v{c}, Vempala, and Vigoda~\cite{SVV}).   The {\em mixing time} is the number of steps, from the worst initial state, to reach within total variation distance $\leq 1/2e$ of the stationary distribution, which in our case
is the Gibbs distribution $\mu$.  

We say that $O(n)$ is the optimal relaxation time and that $O(n\log{n})$ is the optimal mixing time (see Hayes and Sinclair~\cite{HS07} for a matching lower bound for any constant degree graph).  
Here, $n$ denotes the size of the underlying graph.  More generally, we say the Glauber dynamics is rapidly mixing when the mixing time is $\poly(n)$.

We establish  bounds on the  mixing time of the Glauber dynamics by means of  {\em approximate tensorization  inequalities} for the variance  
 of the hard-core model. Interestingly, our
analysis utilizes nothing further than the inductive nature of the tree, e.g., we do not make any assumptions about spatial mixing properties of the 
Gibbs distribution.  As a consequence, the bounds we obtain have no dependence on the maximum degree of the graph. 

To be more specific we derive the following two group of results: We  establish approximate tensorization of variance
of the hard-core model on the tree for all  $\lambda<1.1$. This implies  {\em optimal}  $O(n)$ relaxation time for the Glauber dynamics.   Notably 
this also includes the uniform distribution over independent sets, i.e., $\lambda=1$. 

We can now state our main results.

\begin{theorem} \label{thm:relaxation}
For any $n$-vertex tree, for any $\lambda<1.1$ the Glauber dynamics for sampling $\lambda$-weighted independent sets in the hard-core model has an optimal relaxation time of $O(n)$.
\end{theorem}

We believe the optimal mixing results of  \cref{thm:relaxation} are related to the reconstruction threshold, which we describe now.   
 Consider the complete $\Delta$-regular tree of height $h$; this is the rooted tree where all nodes at distance $\ell<h$ from the root have $\Delta-1$ children and all nodes at distance~$h$ from the root are leaves.   
 We are interested in how the configuration at the leaves affects the configuration at the root.

 Consider fixing an assignment/configuration $\sigma$ to the leaves (i.e., specifying which leaves are fixed to occupied and which are unoccupied), we refer to this fixed assignment to the leaves as a boundary condition~$\sigma$.  Let $\mu_\sigma$ denote the Gibbs distribution conditional on this fixed boundary condition~$\sigma$, and let $p_\sigma$ denote the marginal probability that the root is occupied in~$\mu_\sigma$.
 
 The uniqueness threshold $\lambda_c(\Delta)$ measures the affect of the {\em worst-case} boundary condition on the root.  For all $\lambda<\lambda_c(\Delta)$, all $\sigma\neq\sigma'$, in the limit $h\rightarrow\infty$, we have $p_\sigma=p_\sigma'$; this is known as the (tree) uniqueness region.  In contrast,  for $\lambda>\lambda_c(\Delta)$ there are pairs $\sigma\neq\sigma'$ (namely, all even occupied vs. odd occupied) for which the limits are different; this is the non-uniqueness region.  The uniqueness threshold is at $\lambda_c(\Delta)=(\Delta-1)^{\Delta-1}/(\Delta-2)^{\Delta} = O(1/\Delta)$.
 
 In contrast, the reconstruction threshold $\lambda_r(\Delta)$ measures the affect on the root of a random/typical boundary condition.  In particular, we fix an assignment $c$ at the root and then generate the Gibbs distribution via an appropriately defined broadcasting process. Finally, we fix the boundary configuration~$\sigma$ and ask whether, in the conditional Gibbs distribution $\mu_\sigma$, the root has a bias towards the initial assignment $c$.  The non-reconstruction region $\lambda<\lambda_r(\Delta)$ corresponds to when we cannot infer the root's initial value, in expectation over the choice of the boundary configuration $\sigma$ and in the limit $h\rightarrow\infty$, see Mossel~\cite{Mossel:survey} for a more complete introduction to reconstruction.
 
 The reconstruction threshold is not known exactly but close bounds were established by  Bhatnagar, Sly, and Tetali~\cite{BST16} and Brightwell and Winkler~\cite{BW}; they showed that for constants $C_1,C_2>0$ and sufficiently large $\Delta$:
 $C_1\log^2{\Delta}/\log\log{\Delta}\leq \lambda_r(\Delta)\leq C_2\log^2{\Delta}$, and hence $\lambda_r(\Delta)$ is ``increasing asymptotically'' with $\Delta$ whereas the uniqueness threshold is a decreasing function of~$\Delta$.
  Martin~\cite{Martin} showed that $\lambda_r(\Delta)>e-1$ for all~$\Delta$. 
 As a consequence, we conjecture that \cref{thm:relaxation} holds for all trees for all $\lambda<e-1$, which is close to the bound we obtain.  A slowdown in the reconstruction region is known:
 Restrepo, {\v{S}}tefankovi\v{c}, Vera, Vigoda, and Yang \cite{RSVVY} showed that there are trees for which there is a polynomial slow down for $\lambda>C$ for a constant $C>0$; an explicit  constant $C$ is not stated in \cite{RSVVY} but using the Kesten-Stigum bound one can show $C\approx 28$ (by considering binary trees).

For general graphs the  appropriate threshold is the uniqueness threshold, which is $\lambda_c(\Delta)=O(1/\Delta)$.  
In particular, for bipartite random $\Delta$-regular graphs the Glauber dynamics has optimal mixing in the uniqueness region by Chen, Liu and Vigoda~\cite{CLV21}, and is exponentially slow in the non-uniqueness region by Mossel, Weitz, and Wormald~\cite{MWW07} (see also \cite{GSV16}).
Moreover, for general graphs there is a computational phase transition at the uniqueness threshold: optimal mixing on all graphs of maximum degree $\Delta$ in the uniqueness region~\cite{CLV21,CFYZ22,CE22}, and NP-hardness to approximately count/sample in the non-uniqueness region by Sly~\cite{Sly10} (see also, \cite{GSV16,SS14}). 

There are a variety of mixing results for the special case on trees, which is the focus of this paper.  In terms of establishing optimal mixing time bounds for the Glauber dynamics, previous results only applied to complete, $\Delta$-regular trees.  Seminal work of Martinelli, Sinclair, and Weitz~\cite{MSW03,MSW04} proved optimal mixing on complete $\Delta$-regular trees for all~$\lambda$.  The intuitive reason this holds for all~$\lambda$ is that the complete tree corresponds to one of the two extremal phases (all even boundary or all odd boundary) and hence it does not exhibit the phase co-existence which causes   mixing. 
 As mentioned earlier, \cite{RSVVY} shows that   there is a fixed assignment~$\tau$ for the leaves of the complete, $\Delta$-regular tree so that the mixing time slows down in the reconstruction region; intuitively, this boundary condition~$\tau$ corresponds to the assignment obtained by the broadcasting process.

For more general trees the following results were known.  A classical result of 
Berger, Kenyon, Mossel and Peres~\cite{BKMP}
proves polynomial mixing time for trees with constant maximum degree~\cite{BKMP}. 
A very recent result of Eppstein and Frishberg~\cite{EF22} proved polynomial mixing time $n^{C(\lambda)}$ of the Glauber 
dynamics for graphs with bounded tree-width which includes arbitrary trees, however the polynomial exponent is 
$C(\lambda)=O(1+|\log(\lambda)|)$ for trees; see more recent work of Chen~\cite{Zongchen-SODA} for further improvements. For other combinatorial models, rapid mixing for the Glauber dynamics on trees 
with bounded maximum degree was established for $k$-colorings in~\cite{LMP09} and edge-colorings in~\cite{DHP20}.

Spectral independence is a powerful notion in the analysis of the convergence rate of Markov Chain Monte Carlo (MCMC) algorithms.
 For independent sets on an $n$-vertex graph $G=(V,E)$, spectral independence considers the $n\times n$ pairwise influence matrix ${\mathcal{I}}_G$ where ${\mathcal{I}}_G(v,w)=\Prob_{\sigma\sim\mu}(v\in\sigma\ |\ w\in\sigma)-\Prob_{\sigma\sim\mu}(v\in\sigma\ |\ w\notin\sigma)$; this matrix is closely related to the vertex covariance matrix.  We say that spectral independence holds if the maximum eigenvalue of ${\mathcal{I}}_{G'}$ for all vertex-induced subgraphs $G'$ of $G$ are bounded by a constant.
 Spectral independence was introduced by Anari, Liu, and Oveis Gharan~\cite{ALO20}.  Chen, Liu, and Vigoda~\cite{CLV21} proved that spectral independence, together with a simple condition known as marginal boundedness which is a lower bound on 
 the marginal probability of a valid vertex-spin pair, implies optimal mixing time of the Glauber dynamics for {\em constant-degree graphs}.  This has led to a flurry of optimal mixing results, e.g.,~\cite{CLV21zerofree,BCCPSV22,Liu21,CLMM23,CG23}.

 The limitation of the above spectral independence results is that they only hold for graphs with constant maximum degree~$\Delta$.  
 There are several noteworthy results that achieve 
 a stronger form of spectral independence which establishes optimal mixing for unbounded degree graphs~\cite{CFYZ22,CE22}; however all of these results for general graphs only achieve rapid mixing in the tree uniqueness region which corresponds to $\lambda=O(1/\Delta)$ whereas we are aiming for $\lambda=\Theta(1)$. 

The inductive approach we use to establish approximate tensorization inequalities can also be utilized to establish spectral independence.  
In fact, we show that spectral independence holds for any tree when $\lambda<1.3$, including the case where $\lambda=1$,
see  \cref{sec:SI}.   Applying the results of Anari, Liu, and Oveis Gharan~\cite{ALO20} we obtain a $\poly(n)$ bound on the mixing time, but
with a large constant in the exponent of $n$.
For constant degree trees we obtain the following optimal mixing result by applying the results of Chen, Liu, and Vigoda~\cite{CLV21} (see also~\cite{BCCPSV22,CFYZ22,CE22}).
\begin{theorem}
    \label{thm:SI-mixing}
For all constant $\Delta$, all $\lambda\leq 1.3$, for any tree $T$ with maximum degree $\Delta$, the Glauber dynamics for sampling $\lambda$-weighted independent sets has an optimal mixing time of $O(n\log{n})$.
\end{theorem}

Interestingly, combining the spectral independence results from the proof \Cref{thm:SI-mixing} 
with results of Chen, Feng, Yin, and Zhang~\cite[Theorem~1.9]{chen2024rapid}, we are able to strengthen \Cref{thm:relaxation} by allowing larger fugacities, i.e., 
 $\lambda \leq 1.3$.

\begin{corollary} \label{cor:relaxationBoost}
For any $n$-vertex tree, for any $\lambda\leq 1.3$ the Glauber dynamics for sampling $\lambda$-weighted independent sets in the hard-core model has an optimal relaxation time of $O(n)$.
\end{corollary}

In the next section we recall the key functional definitions and basic properties of variance that will be useful later in the proofs.  In \cref{sec:variance} we prove approximate tensorization of variance which establishes~\cref{thm:relaxation}.  
We establish spectral independence and prove~\cref{thm:SI-mixing} in  \cref{sec:SI}.  

\begin{remark}
\label{rem:mistake}
An earlier version of this paper \cite{arxiv-version} claimed to prove $O(n\log{n})$ mixing time for $\lambda\leq .44$ for any tree (without any constraint on the maximum degree).  There was a mistake in that proof.  In particular, inequality (54) is false.  Moreover, Zongchen Chen pointed out a simple test function $f$ which shows that entropy tensorization and modified log-Sobolev inequality (MLSI) do not hold for the star graph with constants independent of the degree.  
\end{remark}

A recent paper of Chen, Yang, Yin, and Zhang \cite{CYYZ} improves \Cref{cor:relaxationBoost} to obtain optimal relaxation time on trees for $\lambda<e^2$.

\section{Preliminaries}

\subsection{Standard Definitions}
Let  $P$ be the transition matrix of a Markov chain $\{X_t\}$  with a finite state space 
$\Omega$ and equilibrium distribution~$\mu$. For $t\geq 0$ and $\sigma\in \Omega$, let 
$P^{t}(\sigma, \cdot)$ denote the distribution of $X_t$ when the initial state of the chain 
satisfies $X_0=\sigma$.  The  {\em mixing time} of the Markov chain  $\{X_t\}_{t \geq 0}$ is 
defined by
\begin{align}\label{def:MixingTime}
   T_{\rm mix} &=  \max_{\sigma\in \Omega}\min { \left\{t > 0 \mid \Vert P^{t}(\sigma, \cdot) - \mu \Vert_{\rm TV} \leq \frac{1}{2 \mathrm{e}} \right \}} .
\end{align}
The transition matrix $P$ with stationary distribution $\mu$ is called {\em time reversible} if it satisfies the so-called {\em detailed balance relation}, 
i.e., for any $x, y\in\Omega$ we have $\mu(x)P(x,y)=P(y,x)\mu(y)$. 
For $P$ that is time reversible the set of eigenvalues are real numbers and we denote them as 
$1=\lambda_1\geq \lambda_2\geq \ldots \lambda_{|\Omega|}\geq -1$. 
Let $\lambda^*=\max\{|\lambda_2|, |\lambda_{|\Omega|}|\}$, then  
we define the {\em relaxation time} $\Trelax$ by
\begin{align}\label{def:RelaxTime}
\Trelax(P)=\frac{1}{1-\lambda^*}. 
\end{align}
The quantity $1-\lambda^*$ is also known as the {\em spectral gap} of $P$.
We use $\Trelax$ to bound $\Tmix$ by using the following inequality
\begin{align}\label{eq:MixingTimeVsTrelax}
\Tmix(P) \leq \Trelax(P)\cdot \log\left(\frac{2e}{\min_{x\in\Omega}\mu(x)}\right).
\end{align}

\subsection{Gibbs Distributions and Functional Analytic Definitions} 

For a graph $G=(V,E)$ and $\lambda>0$,  let $\mu=\mu_{G,\lambda}$ be the hard-core model on this graph with activity~$\lambda$, 
while let  $\Omega\subseteq 2^V$ be the support of $\mu$, i.e., $\Omega$ are the collection of independent sets of $G$.

For any function $f:\Omega\to\mathbb{R}_{\geq 0}$,  we let $\mu(f)$ is the expected value of $f$ with respect to $\mu$, i.e., 
\begin{align*}
\mu(f)  &= \sum_{\sigma \in \Omega}\mu(\sigma)f(\sigma).
\end{align*}
Analogously, we define variance of $f$ with respect to $\mu$ by 
\[
\mathrm{Var}(f) = \textstyle \mu(f^2)-\left(\mu(f)\right)^2 .
\]
We are also using the following equation for $\mathrm{Var}(f)$, 
\[
\mathrm{Var}(f) = \textstyle \frac12 \sum_{\sigma, \tau\in \Omega} \mu(\sigma)\mu(\tau)\left(f(\sigma)-f(\tau)\right)^2 .
\]

 For any subset $S \subseteq V$, let $\Omega_S$ denote the set of independent sets on $S$.  Then, let $\mu_S$ denote the   marginal of  $\mu$ on $S$; that is, for any $\sigma\in\Omega_S$, we have that
 \[
 \mu_S(\sigma) =\sum_{\eta\in \Omega}\mathbf{1}\{ \eta \cap  S = \sigma \} \mu(\eta). 
 \]
 For any $S\subset V$, any $\tau \in \Omega_{V\setminus S}$, we let $\mu^{\tau}_S$ be the distribution $\mu$ 
conditional on the configuration $\tau$ on $V\setminus S$, and let $\Omega_S^\tau$ to be the support of $\mu^\tau_S$.

For any  $S \subseteq V$, for any  $\tau \in \Omega_{V \setminus S}$,  we define the function 
$f_{\tau}: \Omega^{\tau}_S \to \mathbb{R}_{\geq 0}$  such that    $f_{\tau}(\sigma) = f(\tau \cup \sigma)$ for all  
$\sigma \in \Omega^\tau_S$.  

Let 
\[
\mu_S^\tau(f) = \sum_{\sigma\in \Omega^{\tau}_S}\mu^{\tau}_S(\sigma)f_{\tau}(\sigma).
\]
Let $\mathrm{Var}_S^\tau(f)$ denote the variance of $f_\tau$ with respect 
to the conditional  distribution $\mu^\tau_S$: 
\begin{align}
\mathrm{Var}_S^\tau(f) & = \textstyle \mu^{\tau}_S(f^2)-\left(\mu^{\tau}_S(f)\right)^2 
&=\frac{1}{2}\sum_{\sigma, \eta \in \Omega} 
\frac{\mathbf{1}\{\sigma\setminus S=\tau, \ \eta\setminus S=\tau\} \mu(\sigma)\mu(\eta)}{\left (\sum_{\hat{\sigma}\in \Omega}
\mathbf{1}\{\hat{\sigma}\setminus S=\tau\} \mu(\hat{\sigma})\right)^2}\left(f(\sigma)-f(\eta)\right)^2 . \label{eq:CondVarCompl}
\end{align}

Furthermore, we let 
\begin{equation}  \label{def:local-var}
\mu(\mathrm{Var}_S(f)) = \sum_{\tau \in \Omega_{V \setminus S}} \mu_{V \setminus S}(\tau)\cdot \mathrm{Var}_S^\tau(f) ,
\end{equation}
i.e.,  $\mu(\mathrm{Var}_S(f))$ is the average of   $\mathrm{Var}_S^\tau(f)$ with respect to $\tau$ being distributed as in    $\mu_{V \setminus S}(\cdot)$.
For the sake of brevity, when  $S=\{v\}$, i.e., the set $S$ is a singleton, we  use   $\mu(\mathrm{Var}_v(f))$.

Finally, let 
\begin{equation} \label{eqn:var-mu-f}
\mathrm{Var}(\mu_S(f)) 
 = 
\mu\left( \left(\mu_S(f) \right)^2\right)-
\left(\mu\left( \mu_S(f) \right)\right)^2  
 = 
\Exp_{\tau\sim\mu_{V\setminus S}}\left[ \left(\mu^\tau_S(f) \right)^2\right]-
\left(\Exp_{\tau\sim\mu_{V\setminus S}}\left[ \mu^\tau_S(f) \right]\right)^2,  
\end{equation}
i.e.,  $\mathrm{Var}(\mu_S(f))$ is the variance of   $\mu_S^\tau(f)$ with respect to $\tau$ being distributed as in    $\mu_{V \setminus S}(\cdot)$.

When $X$ is the following two-valued random variable:
\[
X = \begin{cases}
A & \mbox{ with probability $p$} \\
B & \mbox{ with probability $1-p$} ,
\end{cases}
\]
then one formulation for the variance that will be convenient for us is
\begin{equation} \label{eq:bernoulli-variance}
\Var(X) = p(1-p)(A-B)^2 .
\end{equation}

\subsection{Approximate Tensorization of Variance}

To bound the convergence rate of the Glauber dynamics we
consider the approximate tensorization of variance as introduced in~\cite{CMT14}.

We begin with the definition of approximate tensorization of variance.

\begin{definition}[Variance Tensorization]\label{def:var-tensorization}
A distribution $\mu$ with support $\Omega \subseteq \{\pm 1\}^V$ satisfies the approximate tensorisation of 
Variance with constant $C>0$, denoted using the predicate $VT(C)$, if for 
all $f:\Omega\to\mathbb{R}_{\geq 0}$ we have that
\begin{align}\nonumber
\mathrm{Var}(f)\leq C\cdot \sum_{v\in V}\mu\left( \mathrm{Var}_v(f)\right) .
\end{align}
\end{definition}
For a vertex $x$, recall that  $\mathrm{Var}_{x}[f]=  \sum_{\sigma}\mu_{V\setminus \{x\}}(\sigma)\mathrm{Var}^{\sigma}_{x}[f_{\sigma}]$.
%
Variance tensorization $VT(C)$  yields the following bound on the
relaxation time of the Glauber dynamics~\cite{CMT14,Caputo-notes}:  
\begin{equation}\label{eq:relax-VT}
    \Trelax \leq  Cn .
\end{equation}

\subsection{Decomposition of Variance}

We use the following basic decomposition properties for  variance.
The following lemma follows from a decomposition of relative entropy, see~\cite[Lemma 3.1]{CP20} (see also~\cite[Lemma 2.3]{BCCPSV22}).
\begin{lemma}\label{el1}
For any $S\subset V$, for any $f\geq 0$:
\begin{align}
    \label{el1:var}
    \mathrm{Var}(f) & = \mu[\mathrm{Var}_S(f)] + \mathrm{Var}(\mu_S(f)),   
\end{align}
where $\mathrm{Var}(\mu_S(f))$ is defined in \cref{eqn:var-mu-f}.
\end{lemma}

We utilize the following variance factorisation  for product measures, see~\cite[Eqn (4.7)]{Caputo-notes}.

\begin{lemma}     \label{lem:el2}    
Consider $U,W\subset V$ where  $\dist(U,W) \ge 2$.  Then for all $f\geq 0$ we have: 
\begin{align}
    \label{eqn:el2-var}
    \mu[\mathrm{Var}_U(\mu_W(f))] &\leq \mu[\mathrm{Var}_U(f)],
\end{align}
\end{lemma}
On a first account,  the reader might find it challenging to parse the expression $\mu[\mathrm{Var}_U(\mu_W(f)) ]$. 
In that respect, \eqref{def:local-var} and \eqref{eqn:var-mu-f} might be useful. 
Specifically,  $\mu[\mathrm{Var}_U(\mu_W(f)) ]$  is the expectation of
$\mathrm{Var}^{\tau}_U(\mu_W(f)) $ with respect to $\tau$ being distributed as in $\mu_{\bar{U}\cap \bar{W}}$.
Furthermore,  $\mathrm{Var}^{\tau}_U(\mu_W(f)) $ corresponds to   the  variance of   $\mu_W^{\tau,\sigma}(f)$ with respect 
to the configurations $\tau$ at  $V\setminus (U\cup W)$ and $\sigma$ at $U$, while $\tau$ is fixed and 
$\sigma$ is distributed as in    $\mu^{\tau}_{U}(\cdot)$. 
\begin{proof}
We apply~\cite[Eqn (4.7)]{Caputo-notes}, which reaches the same conclusion under the assumptions that
$U \cap W = \emptyset$ and $\mu_U \mu_W = \mu_W \mu_U$.  
In the current context, the reason these conditional expectation operators commute here is that,
because $\dist(U,W) \ge 2$, if we let $S$ be an independent set sampled according to distribution $\mu$,
then the random variables
   $S \cap U$ and $S \cap W$ are conditionally independent given $S \setminus (U \cup W)$.
\end{proof}

\section{Variance Factorization}
\label{sec:variance}

In this section we prove \cref{thm:relaxation}, establishing an optimal bound on the relaxation time for the Glauber dynamics on any tree for $\lambda\leq 1.1$.   We will prove this by establishing variance tensorization, see~\cref{def:var-tensorization}.

Consider a graph $G=(V,E)$ and a collection of fugacities $\lambda_i$ for each $i\in V$.  Throughout this section we will assume that all the fugacities are bounded by $1.1$. Consider the following more general definition of the Gibbs distribution $\mu$ for the hard-core model, where for an independent set $S$, 
 \begin{equation}\label{def:general-gibbs}
 \mu(S)\propto \prod_{i\in S}\lambda_i .
 \end{equation}

Let $T'=(V',E')$ be a tree, let $\{\lambda'_w\}_{w\in V'}$ be a collection of fugacities and let $\mu'$ be the corresponding Gibbs distribution. 
We will establish the following variance tensorization inequality: for all $f':2^{V'}\rightarrow{\mathbb R}$
\begin{equation}\label{vti}
\mathrm{Var}(f') \leq \sum_{x\in V'} F(\lambda'_x) \mu'(\mathrm{Var}_x(f')) ,
\end{equation}
 where $F:\R_{\geq 0}\rightarrow\R_{\geq 0}$ is a function to be determined later (in~\cref{lem:exp-induct}).   We refer to $\mathrm{Var}(f')$ as the ``global'' variance and we refer to $\mu'(\mathrm{Var}_x(f'))$ as the ``local'' variance (of $f'$ at $x$).

We will establish~\eqref{vti} using induction. Let $v$ be a vertex of degree $1$ in $T'$ and let $u$ be the unique neighbor of~$v$.
Let $T=(V,E)$ be the tree by removing $v$ from $T'$, i.e., $T$ is the induced subgraph of $T'$ on $V=V'\setminus\{v\}$.
Let $\{\lambda_w\}_{w\in V}$ be a collection of fugacities where $\lambda_w = \lambda'_w$ for $w\neq u$ and $\lambda_u=\lambda'_u/(1+\lambda'_v)$. 
Let $\mu$ be the hard-core measure on $T$ with fugacities $\{\lambda_w\}_{w\in V}$.

Note that for $S\subseteq V$ we have
\begin{equation}\label{e1}
\mu(S) = \mu'(S) + \mu'(S\cup\{v\}) = \mu'_V(S) . 
\end{equation}
Fix a function $f':2^{V'}\rightarrow\R$. Let $f:2^V\rightarrow\R$  be defined by
\begin{equation}\label{defn:f'}
f(S) = \frac{\mu'(S) f'(S) + \mu'(S\cup\{v\}) f'(S\cup\{v\})}{\mu'(S)  + \mu'(S\cup\{v\})} = \Exp_{Z\sim\mu'}[f'(Z)\mid Z\cap V = S] = \mu'_{v}(f')(S) .
\end{equation}
Note that $f'$ is defined on independent sets of the tree $T'$ and $f$ is the natural projection of $f'$ to the tree~$T$.  
Since $f=\mu'_v(f')$, then by \cref{el1} we have that:
\begin{equation}
\label{var:el1-step}
    \mathrm{Var}(f') = \mu'(\mathrm{Var}_v(f')) + \mathrm{Var}(f) .
\end{equation}

For measure $\mu'$,  when we condition on the configuration at $u$, the configuration at $V\setminus\{u\}$ is independent of that at $\{v\}$.
Hence, from~\cref{eqn:el2-var} for any $x\not\in \{u,v\}$ (by setting $U=\{x\}$ and $W=\{v\}$) we have:

\begin{equation*}
\mu'(\mathrm{Var}_x(f))\leq\mu'(\mathrm{Var}_x(f')) .
\end{equation*}
Since by \eqref{e1} we have $\mu(\mathrm{Var}_x(f))=\mu'(\mathrm{Var}_x(f))$,  the above implies  that
\begin{equation}
    \label{var:el2-step}
\mu(\mathrm{Var}_x(f))\leq\mu'(\mathrm{Var}_x(f')) .
\end{equation}

The following lemma is the main technical ingredient.  It bounds the local variance at $u$ for the smaller graph $T$ in terms of the local variance at $u$ and $v$ in the original graph $T'$.
\begin{lemma}
    \label{lem:exp-induct}
    For $F(x)=1000 (1+x)^2 \exp(1.3x)$ and any $\lambda_v,\lambda_u\in (0,1.1]$ we have:
\begin{equation}\label{miss}
F(\lambda_u) \mu(\mathrm{Var}_u(f)) \leq (F(\lambda'_v) - 1)\mu'(\mathrm{Var}_v(f')) + F(\lambda'_u) \mu'(\mathrm{Var}_u(f')) .
\end{equation}
\end{lemma}

We now utilize the above lemma to prove the main theorem for the relaxation time.  We then go back to prove \cref{lem:exp-induct}.
\begin{proof}[Proof of \cref{thm:relaxation}]
Note \cref{miss} is equivalent to:
\begin{equation}\label{zzzzz}
\mu'(\mathrm{Var}_v(f'))  +
    F(\lambda_u) \mu(\mathrm{Var}_u(f)) \leq F(\lambda'_v)\mu'(\mathrm{Var}_v(f'))
    +  F(\lambda'_u) \mu'(\mathrm{Var}_u(f'))  .
\end{equation}
We can prove variance tensorization by induction as follows:
\begin{align*}
\nonumber
\mathrm{Var}(f') 
&= \mu'(\mathrm{Var}_v(f')) + \mathrm{Var}(f) 
\\
\nonumber
&\leq  \mu'(\mathrm{Var}_v(f')) + \sum_{x\in V} F(\lambda_x) \mu(\mathrm{Var}_x(f))  
\\ 
\nonumber
& \leq 
 \mu'(\mathrm{Var}_v(f')) + F(\lambda_u) \mu(\mathrm{Var}_u(f)) +  \sum_{x\in V\setminus\{u\}} F(\lambda'_x) \mu'(\mathrm{Var}_x(f')) 
 \\
& \le 
F(\lambda'_v)\mu'(\mathrm{Var}_v(f')) +  F(\lambda'_u) \mu'(\mathrm{Var}_u(f')) +  \sum_{x\in V\setminus\{u\}} F(\lambda'_x) \mu'(\mathrm{Var}_x(f')) 
\\
 & =
 \sum_{x\in V'} F(\lambda'_x) \mu'(\mathrm{Var}_x(f')).
\end{align*}
For the first line, we use \cref{var:el1-step}. The second line follows from the inductive hypothesis. For the third line, we use \cref{var:el2-step} and the 
fact that $F(\lambda_x)\leq F(\lambda'_x) $, since $F$ is increasing and $\lambda_x\leq \lambda'_x$. The fourth line follows by using  \cref{zzzzz}.
\end{proof}

Our task now is to prove \cref{lem:exp-induct}.   The following technical inequality will be useful.
\begin{lemma}\label{uin}
Let $p\in [0,1]$. Suppose $s_1,s_2>0$ satisfy $s_1 s_2 \geq 1$. Then for all $A,B,C\in\R$ we have
\begin{equation}\label{usf}
(C-p A - (1-p)B)^2 \leq (1+s_1)(C-A)^2 + (1-p)^2 (1 + s_2) (B-A)^2  .
\end{equation}
\end{lemma}

 \begin{proof}
 Equation~\eqref{usf} is equivalent to   
\begin{equation}\label{usf-B}
 2(1-p)(C-A)(A-B)
 \leq s_1 (C-A)^2 + s_2 (1-p)^2 (B-A)^2  .
 \end{equation}
 We derive the  above by  subtracting $(C-A)^2$ and $(1-p)^2(B-A)^2$ from both sides  of 
  equation~\eqref{usf} and rearranging. 
  
  A simple application of the AM-GM inequality implies that  
 $$2 \sqrt{s_1s_2} (1-p)(C-A)(A-B)
 \leq s_1 (C-A)^2 + s_2 (1-p)^2 (B-A)^2.
 $$
 Then, equation~\eqref{usf-B}  follows  from
 the observation that  the left-hand side of the above inequality is at least 
 $2(1-p)(C-A)(A-B)$, i.e., since $s_1s_2\geq 1$.
 \end{proof}

We can now prove the main lemma.

\begin{proof}[Proof of \cref{lem:exp-induct}]
Our goal is to prove \cref{miss}, let us recall its statement:
\begin{equation}
\tag{\ref{miss}}
F(\lambda_u) \mu(\mathrm{Var}_u(f)) \leq (F(\lambda'_v) - 1)\mu'(\mathrm{Var}_v(f')) + F(\lambda'_u) \mu'(\mathrm{Var}_u(f'))  .
\end{equation}

We will consider each of these local variances $\mu(\mathrm{Var}_u(f))$, $\mu'(\mathrm{Var}_v(f'))$, and $\mu'(\mathrm{Var}_u(f'))$.  
 We will express each of them as a sum over independent sets $S$ of $V'$.  We can then establish~\cref{miss} in a pointwise manner by considering the corresponding inequality for each independent set $S$.

Let us begin by looking at the general definition of the expected local variance $\mu'(\mathrm{Var}_x(f'))$ for any $x\in V'$. 
Applying the definition in~\cref{def:local-var} and then simplifying we obtain the following (a reader familar with the notation can apply~\cref{eq:bernoulli-variance} to skip directly to the last line):
\begin{align}
\nonumber
\lefteqn{
\mu'(\mathrm{Var}_x(f')) }
\\
&= \sum_{S \subseteq V' \setminus \{x\}} \mu'_{V' \setminus \{x\}}(S)\cdot \mathrm{Var}_x^S(f_S) 
\nonumber
\\
&= \sum_{S \subseteq V' \setminus \{x\}} \left( \sum_{T\subseteq \{x\}} 
\mu'(S\cup T)\right)
\left(
\frac12\sum_{T,U\subseteq \{x\}, T\neq U}
\mu_{x}^{'S}(T)\mu_x^{'S}(U)(f'(S\cup T)-f'(S\cup U))^2\right)
\nonumber
\\
&= \sum_{S \subseteq V' \setminus \{x\}} \left( \sum_{T\subseteq \{x\}} 
\mu'(S\cup T)\right)
\left( 
\mu_{x}^{'S}(x)\mu_x^{'S}(\emptyset)(f'(S)-f'(S\cup \{x\}))^2\right)
\nonumber
\\
&= \sum_{S\subseteq V'\setminus\{x\}}
\Big(\mu'(S)+\mu'(S\cup\{x\})\Big) \frac{\mu'(S)\mu'(S\cup\{x\})}{(\mu'(S)+\mu'(S\cup\{x\}))^2}
\Big(f'(S) - f'(S\cup\{x\}) \Big)^2  .
\label{lov}
\end{align}

Notice in~\cref{lov} that the only $S\subset V'\setminus\{x\}$ which contribute are those where $x$ is unblocked (i.e., no neighbor of $x$ is included in the independent set $S$) because we need that $S$ and $S\cup\{x\}$ are both independent sets and hence have positive measure in $\mu'$.

Let us now consider each of the local variances appearing in~\cref{miss} (expressed using carefully chosen summations that will allow us to prove~\eqref{miss} term-by-term in terms of $S$).

Let $Q_1:=\mu(\mathrm{Var}_u(f))$ denote the expected local variance of $f$ at $u$.
We will use~\eqref{lov} (applied to $T$ instead of $T'$); note that only $S$ where $u$ is unblocked (that is, when no neighbor of $u$ is occupied)
contribute to the local variance.  Moreover, we can express the expected local variance of $f$ at $u$ in terms of only those $S$ where $u\in S$.  In particular, consider an independent set $S'$ where $u\notin S'$.  Note that if $u$ is blocked (i.e., $N(u)\cap S'\neq\emptyset$) then the local variance at $u$ is zero for this term.  And for those $S'$ where $u\notin S'$ and $u$ is unblocked then the local variance  has the same contribution
as $S=S'\cup\{u\}$ times $1/\lambda_u$ (since $\mu(S'\cup\{u\})=\lambda_u \mu(S')$). Hence  the expected local variance of $f$ at $u$ is given by
\begin{equation*}
\begin{split}
Q_1 := \mu(\mathrm{Var}_u(f)) = \sum_{S\subseteq V; u\in S} \mu(S) \left(1+\frac{1}{\lambda_u}\right)\frac{\lambda_u}{1+\lambda_u}\frac{1}{1+\lambda_u}
\left(f(S\setminus\{u\})-f(S)\right)^2.
\end{split}
\end{equation*}
We have $f(S)=f'(S)$ (since $u\in S$) and $f(S\setminus\{u\}) = \frac{1}{1+\lambda'_v} f'(S\setminus\{u\}) - \frac{\lambda'_v}{1+\lambda'_v} f'(S\setminus\{u\}\cup\{v\})$.
Plugging these in and simplifying we obtain the following:
\begin{equation}\label{q1}
\begin{split}
Q_1 = \frac{1+\lambda'_v}{1+\lambda'_u+\lambda'_v} \sum_{S\subseteq V; u\in S} \mu(S)
\left(f'(S)-\frac{1}{1+\lambda'_v} f'(S-u) - \frac{\lambda'_v}{1+\lambda'_v} f'(S-u+v)\right)^2.
\end{split}
\end{equation}

We now consider $Q_2:=\mu'(\mathrm{Var}_u(f'))$.   As we did for $Q_1$ we can express $Q_2$ as a sum over indpendent sets $S$ where $u\in S$.  In this case for an independent set $S$ where $u\in S$, consider $S'=S\setminus\{u\}$ and note that the following hold
$$
\mu'(S') = \mu(S') \frac{1}{1+\lambda'_v} = \mu(S) \frac{1}{\lambda'_u} \frac{1}{1+\lambda'_v},
$$
and
$$
\mu'(S) = \mu(S).
$$
Hence, we have the following:
\begin{align}
\nonumber
Q_2 := \mu'(\mathrm{Var}_u(f')) 
& = \sum_{S\subseteq V; u\in S} \mu(S) \left(1+\frac{1}{\lambda'_u}\frac{1}{1+\lambda'_v}\right)\frac{\lambda'_u}{1+\lambda'_u}\frac{1}{1+\lambda'_u}
\left(f'(S-u)-f'(S)\right)^2
\\
& = \left(\lambda'_u + \frac{1}{1+\lambda'_v} \right) \frac{1}{(1+\lambda'_u)^2} \sum_{S\subseteq V; u\in S} \mu(S)  \left(f'(S)-f'(S-u)\right)^2.
\label{q2}
\end{align}

Finally, we consider $\mu'(\mathrm{Var}_v(f'))$, the expected local variance of $f'$ at $v$. We will establish a lower bound which we will denote by $Q_3$ (note, $Q_1$ and $Q_2$ were identities but here we will have an inequality).  

To compute $\mu'(\mathrm{Var}_v(f'))$, the expected local variance of $f'$ at $v$, we need to generate an independent set~$S'$ from $\mu'$. Only those~$S'$ where $v$ is unblocked
(that is where $u$ is not in $S'$) contribute to the local variance. We can generate $S'$ by generating $S$ from $\mu$ (whether
we add or do not add $v$ does not change the contribution to the local variance).  As in~\cref{lov}, we obtain the following:
\begin{align*}
\mu'(\mathrm{Var}_v(f')) 
& =
\sum_{S'\subseteq V;  u\not\in S'} \mu(S') \frac{1}{1+\lambda'_v} \frac{\lambda'_v}{1+\lambda'_v}  \left(f'(S'\cup\{v\})-f'(S')\right)^2  \\
& \geq
\sum_{\substack{S'\subseteq V;  u\not\in S'\\ \textrm{$u$ is unblocked}}} \mu(S') \frac{1}{1+\lambda'_v} \frac{\lambda'_v}{1+\lambda'_v}  \left(f'(S'\cup\{v\})-f'(S')\right)^2  \\
& =
\sum_{S\subseteq V; u\in S} \mu(S) \frac{1}{\lambda'_u} \frac{1}{1+\lambda'_v} \frac{\lambda'_v}{1+\lambda'_v}  \left(f'(S\cup\{v\}\setminus\{u\})-f'(S\setminus\{u\})\right)^2.
\end{align*}

Let $Q_3$ denote the lower bound we obtained above:
\begin{equation}
    \label{q3}
 Q_3 := \frac{1}{\lambda'_u} \frac{1}{1+\lambda'_v} \frac{\lambda'_v}{1+\lambda'_v}  \sum_{S\subseteq V; u\in S} \mu(S) \left(f'(S\cup\{v\}\setminus\{u\})-f'(S\setminus\{u\})\right)^2 
\geq \mu'(\mathrm{Var}_v(f'))  .
\end{equation}

Plugging in~\eqref{q1},~\eqref{q2},~\eqref{q3} we obtain that \cref{miss} follows from the following inequality:
\begin{equation}\label{miss2}
F(\lambda_u) Q_1 \leq (F(\lambda'_v) - 1) Q_3+ F(\lambda'_u) Q_2  .
\end{equation}
We will establish~\eqref{miss2} term-by-term, that is, for each $S$ in the sums of~\eqref{q1},~\eqref{q2},~\eqref{q3}. Fix $S\subseteq V$ such that $u\in S$ and let $A=f'(S-u)$, $B=f'(S-u+v)$, and $C=f'(S)$.
We need to show
\begin{equation}\label{etopr}
\begin{split}
\frac{1+\lambda'_v}{1+\lambda'_u+\lambda'_v}\left(C -\frac{1}{1+\lambda'_v} A -  \frac{\lambda'_v}{1+\lambda'_v}B \right) ^2 F\left(\frac{\lambda'_u}{1+\lambda'_v}\right) \hspace{1.5in}
\\ \leq
\frac{1 + \lambda'_v + \lambda'_u }{ 1+\lambda'_v }  \frac{1}{(1+\lambda'_u)^2}\left(C-A\right)^2 F(\lambda'_u) + \frac{1}{\lambda'_u(1+\lambda'_v)^2}(B-A)^2\left(F(\lambda'_v) - 1\right)  .
\end{split}
\end{equation}
Let $p=1/(1+\lambda'_v)$. We will match~\eqref{usf} to~\eqref{etopr}, by first dividing both sides of~\eqref{etopr} by $\frac{1+\lambda'_v}{1+\lambda'_u+\lambda'_v}F\left(\frac{\lambda'_u}{1+\lambda'_v}\right)$ and then choosing
$$
1+s_1 = \left(\frac{1+\lambda'_u+\lambda'_v}{(1+\lambda'_v)(1+\lambda'_u)}\right)^2 \cdot \frac{F(\lambda'_u)}{F\left(\frac{\lambda'_u}{1+\lambda'_v}\right)}\quad\mbox{and}\quad
1+s_2 = \frac{1+\lambda'_u+\lambda'_v}{\lambda'_u(1+\lambda'_v){\lambda'_v}^2} \cdot \frac{F(\lambda'_v)-1}{F\left(\frac{\lambda'_u}{1+\lambda'_v}\right)}
 .
$$
Note that with this choice of $s_1$ and $s_2$ equations 
\eqref{usf} and \eqref{etopr} are equivalent, and hence  to prove~\eqref{etopr} it
is enough to show $s_1 s_2\geq 1$.

\begin{claim} $s_1s_2\geq 1.$
\label{claim:s1s2}
\end{claim}

This completes the proof of the lemma.
\end{proof}

We use the following lemma to prove \cref{claim:s1s2}
 \begin{lemma}
 \label{lem:eeee1}
 Let $\alpha=1.1$ and $\beta=1.3$. Suppose $x,y\in (0,\alpha]$ are such that $(1+x)y\in [0,\alpha]$. Then
 \begin{equation}\label{zeeee1}
 \left( \exp(\beta xy) - 1\right)
 \left(\frac{(1+x)}{(1+y) y x^2} \exp(\beta(x-y))
 - 1\right)\geq 1.1  .
 \end{equation}
 \end{lemma}

 \begin{proof}
 We will show that for $x,y,(1+x)y\in (0,\alpha]$ we have
 \begin{equation}\label{eeeeX1}
 \frac{(1+x)}{(1+y) x} \exp(\beta(x-y)) \geq 1.15  ,
 \end{equation}
 and
 \begin{equation}\label{eeeeX2}
 \left( \exp(\beta xy) - 1\right)
 \left(\frac{1.15}{xy} 
 - 1\right)\geq 1.1  .
 \end{equation}
 To see that~\eqref{eeeeX1} and~\eqref{eeeeX2} imply~\eqref{zeeee1} note 
  \begin{equation*}
 \left( \exp(\beta xy) - 1\right)
 \left(\frac{(1+x)}{(1+y) y x^2} \exp(\beta(x-y))
 - 1\right)\geq 
  \left( \exp(\beta xy) - 1\right)
 \left(\frac{1.15}{ y x}
 - 1\right)
 \geq 1.1  ,
 \end{equation*}
 where the first inequality follows from~\eqref{eeeeX1} and the second from~\eqref{eeeeX2}.

 Note that 
 the constraints on $x,y$ imply that $y+xy\leq\alpha$ and $xy\leq \alpha y$.
 Hence $xy \leq \alpha/(1+1/\alpha)$. To prove~\eqref{eeeeX2} it is sufficient 
 to show for $z\in [0,\alpha/(1+1/\alpha)]$
 \begin{equation}\label{eeeeX3}
 \left(\exp(\beta z) - 1\right)\left(1.15 - z \right) - 1.1 z \geq 0  .
 \end{equation}
 We have
 $$
 \frac{\partial^2}{\partial z^2}\Big( \left(\exp(\beta z) - 1\right)\left(1.15 - z \right) - 1.1 z\Big) 
 = \exp(\beta z)\beta \left(1.15\beta - \beta x - 2 \right) < 0  .
 $$
 Hence $\left(\exp(\beta z) - 1\right)\left(1.15 - z \right) - 1.1 z$ is concave
 and we only need to check~\eqref{eeeeX3} for the endpoints of the interval; 
 for $z=0$ LHS of~\eqref{eeeeX3} is zero, for $z=\alpha/(1+1/\alpha)$ the
 LHS of~\eqref{eeeeX3} has value larger than $0.005$. This concludes the proof
 of~\eqref{eeeeX2}.

 To prove~\eqref{eeeeX1} note that
 \begin{equation}\label{eeeeX4}
 \frac{\partial}{\partial y} (1+x)\exp(\beta(x-y)) - 1.15(1+y) x
 = -(1+x)\beta\exp(\beta(x-y)) - 1.15x < 0  .
 \end{equation}
 Hence we only need to prove~\eqref{eeeeX1} for $y=\alpha/(1+x)$; this simplifies to showing
 \begin{equation}\label{eeeeX5}
 \exp(1.3(x-1.1/(1+x))) \geq \left(\frac{1.15(2.1+x) x}{(1+x)^2}\right)  .
 \end{equation}
 For $x=0$ and $x=1.1$ we have that~\eqref{eeeeX5} is satisfied (using interval arithmetic). 
 Let 
 $$
 Q(x) := 1.3(x-1.1/(1+x)) - \ln \left(\frac{1.15(2.1+x) x}{(1+x)^2}\right)  .
 $$
 The critical points of $Q(x)$ are roots of 
 $$
 1330 x^4 + 5330x^3 + 8290x^2 + 3733x - 2100  . 
 $$
 Since this polynomial is increasing when $x$ is non-negative, it has exactly one positive real root,
 which is in the interval $[0.30765554,0.30765555]$. 
 The value of $Q(x)$ on both endpoints of the
 interval is at least $0.0032$. The derivative of $Q(x)$ (a rational function) is bounded in 
 absolute value by $1$ on the interval and hence $Q(x)>0$ on the entire interval (in particular
 at the critical point). This proves~\eqref{eeeeX5}.
 \end{proof}

We can now complete the proof of \cref{claim:s1s2}.

 \begin{proof}[Proof of~\cref{claim:s1s2}]
 We will use the following substitution to simplify the expression $F(x)=(1+x)^2 H(x)$. Note that $H(x)=1000\exp(1.3x)=1000\exp(\beta x)$. In terms of $H(x)$ we have
 $$
 1+s_1 = \frac{H(\lambda'_u)}{H\left(\frac{\lambda'_u}{1+\lambda'_v}\right)}\quad\mbox{and}\quad
 1+s_2 = \frac{1+\lambda'_v}{(1+\lambda'_u+\lambda'_v)\lambda'_u}\cdot\frac{(1+\lambda'_v)^2 H(\lambda'_v) - 1}{H\left(\frac{\lambda'_u}{1+\lambda'_v}\right)}  .
 $$
 Let $\lambda'_v=x$ and $\lambda'_u=y(1+x)$. We have
 $$
 1+s_1 = \frac{H(y(1+x))}{H(y)}\quad\mbox{and}\quad
 1+s_2 = \frac{1+x}{y(1+y)x^2}\cdot\frac{ H(x) - \frac{1}{(1+x)^2}}{H(y)}  .
 $$
Recall, our goal is  to show $s_1s_2\geq 1$. First we show that for $x,y\in (0,1.1]$ such that $(1+x)y\leq 1.1$ we have
 \begin{equation}\label{ttqw}
 s_2 = \frac{1+x}{y(1+y)x^2}\cdot\frac{ H(x) - \frac{1}{(1+x)^2}}{H(y)} - 1 
 \geq 
 \frac{999}{1000}\left(
 \frac{1+x}{y(1+y)x^2}\cdot\frac{ H(x)}{H(y)} - 1 
 \right)  .
 \end{equation}
 Note that~\eqref{ttqw} is equivalent to
 \begin{equation}\label{equo}
 \frac{1+x}{x^2} \exp(1.3x)- \frac{1}{x^2(1+x)}
 \geq 
  y(1+y) \exp(1.3y)  .
 \end{equation}
 Note that the RHS of~\eqref{equo} is increasing in $y$ and hence
 it is sufficient to show~\eqref{equo} for the maximal value of $y=1.1/(1+x)$;
 this is equivalent to
 \[
 (1+x)^3 \exp(1.3x)- (1+x)
 \geq 
 1.1x^2 (2.1 + x) \exp(1.43/(1+x))  ,
 \]
 the last inequality follows from
 \[(1+x)^3 \left( 1+1.3x + \frac{1.3^2}{2} x^2 + \frac{1.3^3}{6}x^3 + \frac{1.3^4}{24}x^4\right) - (1+x)
 \geq 
 1.1x^2 (2.1 + x) \exp(1.43)   ,
 \]
 checked using Sturm sequences. This concludes the proof of \eqref{ttqw}. 
 
 Note that~\cref{lem:eeee1} is equivalent to the following
 \begin{equation*}
s_1
 \left(\frac{(1+x)}{y (1+y) x^2} \frac{H(x)}{H(y)}
 - 1\right) = 
 \left( \exp(\beta xy) - 1\right)
 \left(\frac{(1+x)}{(1+y) y x^2} \exp(\beta(x-y))
 - 1\right) 
 \geq 1.1  ,
 \end{equation*}
and hence
$$
s_1\geq  1.1\left(\frac{(1+x)}{y (1+y) x^2} \frac{H(x)}{H(y)}
 - 1\right)^{-1},
$$
 which combined with~\eqref{ttqw} yields $s_1 s_2\geq 1.1(999/1000) > 1$, concluding the proof.
 \end{proof}

 \section{Spectral Independence for Trees}
 \label{sec:SI}

 Let $G=(V,E)$ be a graph. Suppose we have a set of fugacities $\lambda_i$ for each $i\in V$, and consider the corresponding Gibbs distribution.
 Recall that the influence matrix ${\mathcal{I}}_G$ has $u,v$ entry ${\mbox{Inf}}(u\ra v)$, where
 $$
 {\mbox{Inf}}(u\ra v) = \mu(v\in I|u\in I) - \mu(v\in I|u\not\in I)  .
 $$
 Let $u,v$ be any two vertices in a graph. We have
 \begin{equation}\label{infbound}
 \big|{\mbox{Inf}}(v\ra u)\big| \leq \max\{\mu(u\in I|v\in I), \mu(u\in I|v\not\in I)\}\leq\frac{\lambda_u}{1+\lambda_u}  .
 \end{equation}
 We have
 \begin{equation}\label{e3o}
 {\mbox{Inf}}(v\ra u) \mu(v\in I)\mu(v\not\in I)= {\mbox{Inf}}(u\ra v) \mu(u\in I) \mu(u\not\in I)  .
 \end{equation}
 Let $D$ be diagonal matrix with entries $\mu(v\in I)\mu(v\not\in I)$. Equation~\eqref{e3o} means
 that $D {\mathcal{I}}_G$ is symmetric. That also means that $D^{1/2} {\mathcal{I}}_G D^{-1/2}$ is symmetric (since
 it is obtained from $D {\mathcal{I}}_G$ by multiplying by the same diagonal matrix on the left and on
 the right). The $u,v$ entry in $D^{1/2} {\mathcal{I}}_G D^{-1/2}$ is
 \begin{equation}\label{esym}
 \begin{split}
 {\mbox{Inf}}(u\ra v) \left(\mu(v\in I) \mu(v\not\in I)\right)^{-1/2} \left(\mu(u\in I) \mu(u\not\in I)\right)^{1/2} = \\
 {\mbox{Inf}}(v\ra u) \left(\mu(u\in I) \mu(u\not\in I)\right)^{-1/2} \left(\mu(v\in I) \mu(v\not\in I)\right)^{1/2}  ,
 \end{split}
 \end{equation}
 which is equal to $\pm\sqrt{{\mbox{Inf}}(u\ra v){\mbox{Inf}}(v\ra u)}$ (take geometric mean of the sides of~\eqref{esym}).
 We will call $M = D^{1/2} {\mathcal{I}}_G D^{-1/2}$ the symmetrized influence matrix (since it is similar to the influence matrix, it has the same spectral
 radius).

 We will prove  the following result.
 \begin{lemma}\label{lep}
 For any forest $T$ with fugacities in $ [0,1.3]$ the spectral radius of the influence matrix of the hard-core model on $T$ is bounded by 10000. 
 \end{lemma}
 We will prove~\cref{lep} by induction; we will prove the following strengthened statement.
 \begin{lemma}\label{lep2}
 For any forest $T$ with fugacities in $[0,1.3]$ the symmetrized
 influence matrix $M$ of the hard-core model on $T$ satisfies  
 $$
 M\preceq{\mathrm{diag}}\Big(\frac{400}{1.3 - \lambda_v^{0.78}}, \ v\in V \Big)  .
 $$
 \end{lemma}

\begin{proof}
 Let $T=(V,E)$ be a forest. Let $v$ be a vertex of degree $1$ in $T$. Let $u$ be the neighbor of $v$ in $T$. Let $T'=(V',E')$ be $T$ with $v$ removed.
 Let $\lambda'_i = \lambda_i$ for $i\in V'\setminus\{u\}$. Let $\lambda'_u = \lambda_u / (1+\lambda_v)$. Let $\mu'$ be the hard-core model on
 $T'$ with fugacities $\{\lambda'_i\}_{i\in V'}$. Note that $\mu'$ is the same as the marginalization
 of $\mu$ to $V'$, that is, for an independent set $I$ of $T'$ we have
 \begin{equation}\label{e2}
 \mu'(I) = \sum_{J\supseteq I} \mu(J)  ,
 \end{equation}
 where $J$ ranges over independent sets of $T$ that contain $I$.

 Let $M$ be the symmetrized influence matrix for $\mu$ and let $M'$ be the symmetrized influence matrix for $\mu'$. Note that~\eqref{e2} implies
 that $M'$ is a submatrix of $M$ (removing the column and row of $M$ corresponding to vertex $v$ yields $M'$).

It is standard to show, e.g., see \cite[Lemma B.3]{ALO20},  that
 \begin{equation}\label{e3}
 {\mbox{Inf}}(v\ra w) = {\mbox{Inf}}(v\ra u) {\mbox{Inf}}(u\ra w)  ,
 \end{equation}
 and
 \begin{equation}\label{e4}
 {\mbox{Inf}}(w\ra v) = {\mbox{Inf}}(w\ra u) {\mbox{Inf}}(u\ra v)  .
 \end{equation}
 Let $\tau^2={\mbox{Inf}}(v\ra u){\mbox{Inf}}(u\ra v)$. Note that, using~\eqref{infbound}, we have
 \begin{equation}\label{tau}
 \begin{split}
 \tau^2 =  {\mbox{Inf}}(u\ra v) {\mbox{Inf}}(v\ra u) \leq \frac{\lambda_v}{1+\lambda_v} \frac{\lambda_u}{1+\lambda_u}  .
 \end{split}
 \end{equation}
 From~\eqref{e3} and~\eqref{e4} we have that $M$ and $M'$ have the following form ($Q$ is a matrix and $z$ is a vector):
 $$
 M'=\left(
     \begin{array}{cc}
       Q & z \\
       z^T & 0 \\
     \end{array}
   \right)
 \quad\mbox{and}\quad
 M =\left(
   \begin{array}{ccc}
     Q & z & \tau z \\
     z^T & 0 & \tau \\
     \tau z^T & \tau & 0 \\
   \end{array}
 \right)  .
 $$
 Let $F$ be an increasing function on $[0,1.3]$. We will take $F(x)=1/H(x)$ where $H(x)=(a-x^b)/400$ and $a=1.3$, $b=0.78$.
 (We will keep $a$ and $b$ as variables to elucidate the choice of $a$ and $b$ in the proof below.)
Suppose that we know $M'  \preceq {\mathrm{diag}}(F(\lambda'_i))$ and we want to conclude $M\preceq {\mathrm{diag}}(F(\lambda_i))$.
 Let $W$ be a diagonal matrix with entries $F(\lambda_i)$ for $i\in V\setminus\{u,v\}$.  We can restate our goal as follows.
 $$
 \mbox{KNOW:}\left(
     \begin{array}{cc}
       W - Q & -z \\
       -z^T & F\Big(\frac{\lambda_u}{1+\lambda_v}\Big) \\
     \end{array}
   \right)\succeq 0
 \quad\quad
 \mbox{WANT:}\left(
   \begin{array}{ccc}
     W - Q & -z & -\tau z \\
     -z^T & F(\lambda_u) & -\tau \\
     -\tau z^T & -\tau & F(\lambda_v) \\
   \end{array}
 \right)\succeq 0  .
 $$
 Since $F$ will be an increasing function, we have $F(\lambda_u / (1+\lambda_v))<F(\lambda_u)$.

 The condition we want is equivalent (by applying row and column operations, specifically subtracting $u$-th row/column times $\tau$ from the $v$-th row/column) to
 $$\left(
   \begin{array}{ccc}
     W - Q & -z & 0 \\
     -z^T & F(\lambda_u) & -\tau(1+F(\lambda_u))  \\
     0 & -\tau(1+F(\lambda_u)) & F(\lambda_v) + 2\tau^2 + \tau^2 F(\lambda_u)\\
   \end{array}
 \right)\succeq 0  .
 $$
 If we show
 \begin{equation}\label{zzzz}
 \left(
   \begin{array}{cc}
     F(\lambda_u) - F(\lambda_u / (1+\lambda_v)) & -\tau(1+F(\lambda_u)) \\
     -\tau(1+F(\lambda_u)) & F(\lambda_v) + 2\tau^2 + \tau^2 F(\lambda_u) \\
   \end{array}
 \right)\succeq 0 , 
 \end{equation}
 then the conclusion will follow (adding the ``know'' positive semidefinite matrix and~\eqref{zzzz} (expanded with zeros to match dimensions)
 yields that the ``want'' matrix is positive semidefinite). Equation~\eqref{zzzz} is equivalent to checking if the determinant is positive (since
 the entry in the first row/column is positive), that is, we need to check
 $$
 \left(F(\lambda_u) - F(\lambda_u / (1+\lambda_v))\right) F(\lambda_v) > \tau^2 \left(1 + F(\lambda_u / (1+\lambda_v))(2 + F(\lambda_u))  \right)  .
 $$
 Hence it is sufficient (using~\eqref{tau}) to show
 $$
 \left(F(\lambda_u) - F\Big(\frac{\lambda_u}{1+\lambda_v}\Big)\right) F(\lambda_v) > \frac{\lambda_v}{1+\lambda_v} \frac{\lambda_u}{1+\lambda_u} \left(1 + F\Big(\frac{\lambda_u}{1+\lambda_v}\Big)(2 + F(\lambda_u))  \right)  .
 $$
 Letting $H(x)=1/F(x)$ the condition becomes the following
 \begin{equation}\label{ooop}
 \left( H\Big(\frac{\lambda_u}{1+\lambda_v}\Big) - H(\lambda_u)\right)  > H(\lambda_v) \frac{\lambda_v}{1+\lambda_v} \frac{\lambda_u}{1+\lambda_u}
 \left( H\Big(\frac{\lambda_u}{1+\lambda_v}\Big)H(\lambda_u) + 2H(\lambda_u) + 1 \right)  .
 \end{equation}
 We are going to search for $H$ that is 1) bounded from above by $1/300$, 2) bounded away from $0$ on $[0,1.3]$ and 3) that satisfies
 \begin{equation}\label{zzz}
 \left( H\Big(\frac{\lambda_u}{1+\lambda_v}\Big) - H(\lambda_u)\right)  > H(\lambda_v) \frac{\lambda_v}{1+\lambda_v} \frac{\lambda_u}{1+\lambda_u} (1+1/100)  .
 \end{equation}
 Note that such $H$ will also satisfy~\eqref{ooop} since $H\Big(\frac{\lambda_u}{1+\lambda_v}\Big)H(\lambda_u) + 2H(\lambda_u)\leq 1/100$ (here the choice of $1/100$
 is arbitrary; we just need something sufficiently small).

 We will search for $H$ of the following form: $H(x) \propto a - x^b$.  Note that~\eqref{zzz} is invariant under scaling of $H$ and hence satisfying the upper bound
 of $1/300$ can be achieved by scaling of $H$. Ultimately the price we will pay for $H$ being small is that $F$ is big and hence we get a weaker upper bound on
 the spectral radius of $M$; we do not optimize the constants at this point. Consider the following function $L$ (obtained as LHS of~\eqref{zzz} divided by RHS of~\eqref{zzz}, excluding
 the constant term $1+1/100$):
 \begin{equation}\label{zzz2}
 L(\lambda_u,\lambda_v):=  \frac{\lambda_u^b (1 + \lambda_u)}{\lambda_u} \frac{ (1+\lambda_v)^b - 1 }{(a - \lambda_v^b) \lambda_v  (1+\lambda_v)^{b-1}}  .
 \end{equation}
 The minimum of the first term in~\eqref{zzz2} is attained for $\lambda_u=(1-b)/b$ and hence setting $b=.78$ we have that:
 \begin{equation}\label{first111}
 \frac{\lambda_u^b (1 + \lambda_u)}{\lambda_u} \geq \frac{(1-b)^{b-1}}{b^b} \geq \frac{1.32}{.78}.
 \end{equation} 
 For the second
 term in~\eqref{zzz2} by setting $a=1.3$ we have the following:
 \begin{equation}\label{zzz3}
 \frac{ (1+\lambda_v)^b - 1 }{(a - \lambda_v^b) \lambda_v  (1+\lambda_v)^{b-1}}\geq
 \frac{ (1+\lambda_v)^b - 1 }{a \lambda_v  (1+\lambda_v)^{b-1}} \geq \frac{b}{a} = \frac{.78}{1.3}.
 \end{equation}
 The second inequality in~\eqref{zzz3} is obtained  from the following inequality which is valid for $b\in [0,1]$ and $x\geq 0$:
 \begin{equation}\label{ber}
 (1+x)^{b} - 1\geq b x (1+x)^{b-1}.
 \end{equation}
For the above, it suffices to prove that
\begin{align}
 (1+x)^{b} - b x (1+x)^{b-1}\geq 1.
\end{align}
We have that for $b\leq 1$:
\begin{align*}
\nonumber (1+x)^{b} - b x (1+x)^{b-1}&=(1+x)^{b-1}(1 - (b-1)x) \\
\nonumber
 &\geq (1+x)^{b-1}(1+x)^{-(b-1)} & \mbox{by Bernoulli's inequality}
 \\
 & \geq 1.
\end{align*}

Finally, plugging in~\eqref{first111} and~\eqref{zzz3} into~\eqref{zzz2} we obtain:
 \begin{equation}\label{fiz}
 L(\lambda_u,\lambda_v)\geq \frac{1.32}{1.3}>1.01.
 \end{equation}
 Equation~\eqref{fiz} implies that~\eqref{zzz} is satisfied. Recall that the statement of the lemma assumed that the fugacities
 are in the interval $[0,1.3]$. For $\lambda\in [0,1.3]$ we have
 \begin{equation}
 H(\lambda) = \frac{a - \lambda^b}{400} \geq \frac{1}{10000} \quad\mbox{and}\quad H(\lambda) = \frac{a - \lambda^b}{400}\leq \frac{1}{300}  .
 \end{equation}
 This implies that $H$ satisfies~\eqref{ooop} and hence for
 $$
 F(\lambda) = \frac{1}{H(\lambda)} \leq 10000,
 $$
we have~\eqref{zzzz} and this completes the proof of~\cref{lep2} by induction.
\end{proof}

We can now complete the proof for spectral independence.

\begin{proof}[Proof of \cref{thm:SI-mixing}]
We apply Theorem 1.12 of \cite{CLV21} which says that if for all pinnings we have $\eta$-spectral independence  and $b$-marginally boundedness then the mixing time of the Glauber dynamics is $C(\eta,\Delta,b)n\log{n}$.  A pinning refers to a fixed configuration $\tau$ on a subset $S$ of vertices; for the hard-core model a pinning of a tree $T$ corresponds to the hard-core model on a forest which is an induced subgraph.  Hence, \cref{lep} implies that $\eta$-spectral independence holds for all pinnings with $\eta=10000$.  The condition $b$-marginally
boundedness (see Definition 1.9 in~\cite{CLV21}) says that for every pinning $\tau$ on a subset $S\subset V$, for every vertex $v\notin S$, for every assignment to $v$ denoted as $\sigma(v)$ which has positive probability in the conditional Gibbs distribution $\mu_\tau$, then the marginal probability is lower bounded as $\mu_\tau(\sigma(v))\geq b$.  This holds for $b\geq \min\{1,\lambda\}/[\lambda+(1+\lambda)^\Delta]$.  Hence, \cite[Theorem 1.12]{CLV21} implies \cref{thm:SI-mixing}.
\end{proof}

\section{ Proof of \Cref{cor:relaxationBoost}}

We prove \Cref{cor:relaxationBoost} by combining  the spectral independence result in \Cref{lep} and \Cref{thm:relaxation},
while we utilise  \cite[Theorem~1.9]{chen2024rapid}.

In \cite{chen2024rapid},  they introduce the notion of ``complete $\eta$-spectral independence". 
Let $\mu$ be the hard-core model on graph $G=(V,E)$ with  fugacity $\lambda$. For  $\eta>0$, 
complete $\eta$-spectral independence  for $\mu$ corresponds to  the following condition: For any induced subgraph $G'$ of $G$, for the  hard-core model  $\mu'$ on $G'$ 
such that each vertex $v\in V$ has fugacity $\lambda'_v\leq \lambda$, the corresponding influence matrix $\mathcal{I}_{G'}$ has   spectral 
radius at most $\eta$.

\Cref{lep} is equivalent to complete $\eta$-spectral independence for all $\lambda\leq 1.3$ with $\eta\leq 10000$.
We can now prove \Cref{cor:relaxationBoost}.

\begin{proof}[Proof of \Cref{cor:relaxationBoost}:]
For the forest $T$, let $\mu$ be the hard-core model with fugacity $\lambda\leq 1.3$.   Also, let $\gamma$ be the spectral gap for Glauber dynamics on
$\mu$.  \Cref{cor:relaxationBoost} follows by showing  that $\gamma=\Omega(n^{-1})$. 

First, note that  if $\lambda<1.1$, then \Cref{thm:relaxation} already implies $\gamma=\Omega(n^{-1})$. We now focus on the case where $\lambda\in [1.1,1.3]$.

For the same forest $T$,  let $\widehat{\mu}$ be the hard-core model on $T$ with fugacity $\widehat{\lambda}=1$.   Let  $\widehat{\gamma}$ be the spectral
gap for Glauber dynamics on $\widehat{\mu}$. From \Cref{thm:relaxation},  we have that   
\begin{align}\label{eq:SPsmall}
\widehat{\gamma}=\Omega(n^{-1}).
\end{align}
\Cref{lep} and \cite[Theorem~1.9]{chen2024rapid} imply the following relation between $\gamma$ and $\widehat{\gamma}$:
for $\theta=\lambda$ and $\eta \leq 10000$, we have
\begin{align}
\gamma &\geq \left( \frac{\theta}{2}\right)^{2\eta+7}  \widehat{\gamma} 
\  > \left( \frac{1}{2}\right)^{20007}  \widehat{\gamma} \label{eq:SPLargeVsSmall}.
\end{align}
From  \eqref{eq:SPsmall} and  \eqref{eq:SPLargeVsSmall},  we get that $\gamma=\Omega(n^{-1})$, for any $\lambda\in [1.1,1.3]$.

From the above we have established  \Cref{cor:relaxationBoost}.
\end{proof}

\section*{Acknowledgements}

We thank Zongchen Chen for pointing out the counterexample to entropy tensorization as discussed in \Cref{rem:mistake}, and the anonymous referee for pointing out \Cref{cor:relaxationBoost}.

\bibliographystyle{alpha}
\bibliography{refs}

\end{document}